\documentclass[11pt]{article}
\usepackage{snapshot}
\usepackage{fullpage}

\usepackage{amsfonts}
\usepackage{amssymb}
\usepackage{amsthm}
\usepackage{amsmath}
\usepackage{stmaryrd}
\usepackage{tipa}
\usepackage{times}
\usepackage{mathrsfs}
\usepackage{latexsym}
\usepackage{mdwlist}

\newcommand{\vol}{\textrm{vol}}
\newcommand{\sgn}{\textrm{sgn}}
\newcommand{\Tr}{\textrm{Tr}}
\newcommand{\LL}{{\bf \mathcal{L}}}
\newcommand{\supp}{\textrm{supp}}
\newcommand{\p}{\textbf{p}}
\newcommand{\q}{\textbf{q}}
\newcommand{\dd}{\textbf{d}}
\newcommand{\rr}{\textbf{r}}
\newcommand{\qq}{\tilde{\textbf{q}}}
\newcommand{\vv}{\textbf{v}}
\newcommand{\1}{\textbf{1}}
\newcommand{\w}{\textbf{w}}
\newtheorem{definition}{Definition}
\newtheorem{lemma}{Lemma}
\newtheorem{theorem}{Theorem}
\newtheorem{proposition}{Proposition}
\newtheorem{fact}{Fact}

\newtheorem{corollary}{Corollary}

\providecommand{\abs}[1]{\lvert#1\rvert} \providecommand{\norm}[1]{\lVert#1\rVert}

\title{Detecting and Characterizing Small Dense Bipartite-like Subgraphs by the Bipartiteness Ratio Measure
}
\author{Angsheng Li\footnote{State Key Laboratory of
Computer Science, Institute of Software, Chinese Academy of Sciences}\\ angsheng@ios.ac.cn
\and Pan Peng\footnote{State Key Laboratory of
Computer Science, Institute of Software, Chinese Academy of Sciences; Department of Computer Science,
Technische Universit{\"a}t Dortmund}\\
pan.peng@tu-dortmund.de}
\date{}

\begin{document}
\maketitle

\begin{abstract}
We study the problem of finding and characterizing subgraphs with small \textit{bipartiteness ratio}. We give a bicriteria approximation algorithm \verb|SwpDB| such that if there exists a subset $S$ of volume at most $k$ and bipartiteness ratio $\theta$, then for any $0<\epsilon<1/2$, it finds a set $S'$ of volume at most $2k^{1+\epsilon}$ and bipartiteness ratio at most $4\sqrt{\theta/\epsilon}$. By combining a truncation operation, we give a local algorithm \verb|LocDB|, which has asymptotically the same approximation guarantee as the algorithm \verb|SwpDB| on both the volume and bipartiteness ratio of the output set, and runs in time $O(\epsilon^2\theta^{-2}k^{1+\epsilon}\ln^3k)$, independent of the size of the graph. Finally, we give a spectral characterization of the small dense bipartite-like subgraphs by using the $k$th \textit{largest} eigenvalue of the Laplacian of the graph.
\end{abstract}

\section{Introduction}
We study the problem of finding subgraphs with small \textit{bipartiteness ratio}. Let $G=(V,E)$ be an undirected graph. Let $L,R$ be two disjoint vertex subsets and $U:=L\cup R$. The bipartiteness ratio of $L,R$ is defined as
\begin{eqnarray}
\beta(L,R)=\frac{2e(L)+2e(R)+e(U,\bar{U})}{\vol(U)},\label{def:bipar}
\end{eqnarray}
where $e(L),e(U,\bar{U})$ denote the number of edges in $L$ and the number of edges leaving from $U$ to the rest of the graph, respectively; and $\vol(U)$, called the volume of $U$, is defined to be the sum of degrees of vertices in $U$. The concept of bipartiteness ratio was recently introduced and used as a subroutine to designing approximation algorithms for Max Cut problem by Trevisan~\cite{Tre09:maxcut}. In particular, Trevisan showed that this combinatorial object has close relation to the largest eigenvalue of the \textit{Laplacian} of $G$, the Goemans-Williamson Relaxation and graph sparsification.

Another motivation of studying bipartiteness ratio is that it can be considered as a quality of \textit{dense bipartite-like subgraphs}, which in turn characterize the communities (or clusters) in Web graphs~\cite{Pen12:local,KRRT99:trawling}. More specifically, a dense bipartite-like subgraph is a pair of disjoint vertex subsets $L,R$ such that `most' of the edges involving the vertices in $L\cup R$ lie between $L$ and $R$. Equivalently, we say that $L,R$ form a dense bipartite subgraph if `few' edges lie totally in $L$ or $R$, or leaving $L\cup R$ to the rest of the graph. The latter formulation turns out to be well captured by the bipartiteness ratio measure of $L,R$, as in Definition~(\ref{def:bipar}), the numerator involves all the edges that are \textit{not} between $L$ and $R$, and the dominator involves all the edges incident to $L\cup R$. It is intuitive that the smaller the bipartiteness, the more likely it behaves like a dense bipartite subgraph. In real applications, we suggest first detecting sets with small bipartiteness ratio (using the algorithms below) and then combining some heuristic algorithms (such as the prune-filter technique used in \cite{KRRT99:trawling}) to process the found sets and to better exploit the community structure of the Web graph.

Thus, we will use the bipartiteness ratio (abbreviated as \textit{B-ratio}) as a measure of a set being dense bipartite-like. We want to extract subgraphs with small B-ratio, which corresponds to good cyber-communities. Furthermore, we are interested in finding \textit{small} communities, which generally contains more interesting information than large communities and may be more substantial in large scale networks. For example, Leskovec et al. investigate the community structure of many real networks by the \textit{conductance} measure~\cite{LLDM09:community,LLM10:empirical}, and they argue that large networks may have a core-periphery structure, where the periphery is composed of easily separable small communities and the nodes in the expander-like core are so intermingled that it is much harder to extract large communities (if exist) from it.

In order to make our algorithm practical, we would like to design a local algorithm to extract subgraphs with small B-ratio. A local algorithm, introduced by Spielman and Teng~\cite{ST04:linear}, is one that given as input a vertex, it only explores a small portion of the graph and finds a subgraph with good property, which has found applications in graph sparsificasion, solving linear equations~\cite{Spi10:spectral}, and designing near-linear time algorithms~\cite{Ten10:laplacian}. Local algorithms have also shown to be both effective and efficient on real network data~(e.g, \cite{LLDM09:community,LLBSB09:protein}).

\subsection{Our Results}
We give approximation, local algorithms and spectral characterization of finding the small subgraphs with small B-ratio, as we argued above, with the goal of extracting small cyber-communities. In the following, we will use the terminology of small dense bipartite-like subgraphs to indicate small subgraphs with small B-ratio.

We first give a bicriteria approximation algorithm for finding the small dense bipartite-like subgraph.

\begin{theorem}~\label{thm:smallBR}
Assume that $G$ has a set $U=(L,R)$ such that $\beta(L,R)\leq \theta$ and $\vol(U)\leq k$, where $\theta<1/4$ and $k>4$, then for any $0<\epsilon<1/2$, there exists an algorithm \verb|SwpDB|$(G,k,\theta,\epsilon)$ that runs in polynomial time and finds a set $(X,Y)$ such that $\vol(X\cup Y)\leq 2k^{1+\epsilon}$, and $\beta(X,Y)\leq 4\sqrt{\theta/\epsilon}$.
\end{theorem}

Note that the approximation ratio does not depend on the size of the graph, since the algorithm is based on a spectral characterization of the B-ratio of the graph given by Trevisan~\cite{Tre09:maxcut}~(see Lemma~\ref{lemma:tre}), which is analogous to the Cheeger's inequality for conductance (see more discussions below).

By incorporating a truncation operation we are able to give a \textit{local algorithm} for the dense bipartite subgraphs.
\begin{theorem}~\label{thm:local}
If there is a subset $U=(L,R)$ of volume $\vol(U)\leq k$ and B-ratio $\beta(L,R)\leq \theta$, where $\theta<1/12$ and $k>2560000$,then there exists a subgraph $U_\theta\subseteq U$ satisfying that $\vol(U_\theta)\geq \vol(U)/2$ and that if $v\in U_\theta$, then for any $0<\epsilon< 1/2$, there exists a local algorithm \verb|LocDB|$(G,v,k,\theta,\epsilon)$ finds a subgraph $(X,Y)$ of volume $O(k^{1+\epsilon})$ and B-ratio $O(\sqrt{\theta/\epsilon})$. Furthermore, the running time of \verb|LocDB| is $O(\epsilon^2\theta^{-2}k^{1+\epsilon}\ln^3k)$.
\end{theorem}

We remark that in both theorems, we can give alternative tradeoff on the bounds of parameters $k$ and $\theta$. For example, in Theorem~\ref{thm:local}, we can require $\theta<0.03$ and $k>11000$ instead (see the proof of the theorem).

Note that the local algorithm runs in time independent of the size of the graph and is sublinear time (in the size of the input graph, denoted as $n$) when the size of the optimal set is sufficiently smaller than $n$ and the approximation ratio of the algorithm is almost optimal in that it almost matches the guarantee of Trevisan's spectral inequality for the B-ratio. This algorithm also improves the work of the second author~\cite{Pen12:local}, who gave a local algorithm for B-ratio guaranteeing that the output set has volume at most $O(k^2)$ and B-ratio at most $O(\sqrt{\theta})$.

Finally, as an application of the algorithm \verb|SwpDB|, we give a spectral characterization of the small dense bipartite subgraph by relating the $k$th largest eigenvalue of the Laplacian of $G$ to the B-ratio of some subsets with volume at most $O(2|E|/k^{1-\epsilon})$. More specifically, if we let $\lambda_0\leq \lambda_1\leq\cdots\leq\lambda_{n-1}$ be the eigenvalues of the Laplacian matrix $\LL$ of the graph $G$, and define the \textit{dense bipartite profile} of the graph as
$$\beta(k):=\min_{\substack{L,R: L\cap R=\emptyset\\ \vol(L\cup R)\leq k}}\beta(L,R),$$
then our spectral characterization implies that
\begin{eqnarray}
\beta(\vol(G)/k^{1-\epsilon})\leq O(\sqrt{(2-\lambda_{n-k})\log_kn}). ~\label{eqn:dbspectra}
\end{eqnarray}

\subsection{Our Techniques}
Our approximation algorithm is based on Trevisan's spectral characterization of the B-ratio $\beta(G)$ of the graph, which is the minimum B-ratio of all possible disjoint vertex subsets $L,R$, that is, $\beta(G)=\beta(\vol(G))$.
Recall that $\lambda_0\leq \lambda_1\leq\cdots\leq\lambda_{n-1}$ are the eigenvalues of $\LL$. Instead of working directly on $\LL$, we study a closely related matrix $M$, which we call the \textit{quasi-Laplacian}, that has the same spectra as $\LL$. Let $\vv_0,\vv_1,\cdots,\vv_{n-1}$ be the corresponding eigenvectors of $M$. Trevisan showed that if $\lambda_{n-1}\geq 2-2\theta$, then by a simple \textit{sweep process} over the largest eigenvector $\vv_{n-1}$, we can find a pair of subsets $X,Y$ with B-ratio at most $2\sqrt{\theta}$. On the other hand, it is well known that the largest eigenvector $\vv_{n-1}$ can be computed fast by the power method, which starts with a ``good'' vector $\q_0$ and iteratively multiplies it by $M$ to obtain $\q_t$, and outputs $\q_T$ by choosing proper $T$. Hence, the power method combined with the sweep process can find a subset with B-ratio close to $\beta(G)$. However, such a method does not give a useful volume bound on the output set.

In order to find \textit{small} dense bipartite subgraphs, we sweep each of the vector $\q_t$ and characterize $\q_t$ in terms of the minimum of B-ratio of all the small sweep sets (the sets found in the sweep process) encountered in all the $T$ iterations. This is done by a \textit{potential function $J(\p,x)$}, which has a nice convergence property that for general vector $\p$ and some $x$, $J(\p M,x)$ can be bounded by a function of $J(\p,x')$ and the B-ratio of the some sweep set~(see Lemma~\ref{lemma:converge}). Using this property, we show inductively that if we choose $\q_0=\chi_v$ for some vertex $v\in V$, $J(\q_t,x)$ can be upper bounded by a function in $t,K$ and the minimum B-ratio of all the sweep sets of volume at most $K$ for all $t\leq T$~(see Lemma~\ref{lemma:upper}). On the other hand, if the graph contains a small dense bipartite subgraph $L,R$ of volume at most $k$, we prove that the potential function also increases quickly in terms of $t$ and $\beta(L,R)$~(see Lemma~\ref{lemma:lower}), which will lead to the conclusion that at least one of the sweep set with volume at most $K$ has B-ratio ``close'' to $\beta(L,R)$ by choosing proper $K$ in terms of $k$ and the starting vertex $v$.

To give local algorithms that run in time independent of the size of the graph, we need to keep the support size of the vectors $\q_t$ small in each iteration. This is done by a truncation operation of a vector that only keeps the elements with large absolute vector value. Let $\qq_0=\chi_v$ and iteratively define $\qq_t$ to be the truncation vector of $\qq_{t-1}M$. We show that both upper bound and lower bound on $J(\q_t,x)$ still approximately holds for $J(\qq_t,x)$, and thus prove the correctness of our local algorithm which sweeps all the vectors $\qq_t$ instead of $\q_t$.

Finally, we use a simple trace lower bound to serve as the lower bound for $J(\q_t,x)$ and obtain the spectral characterization of the dense bipartite profile.

\subsection{Related Works}
Our work is closely related to a line of research on the \textit{conductance} of a set $S$, which is defined as
$$\phi(S)=\frac{e(S,\bar{S})}{\min\{\vol(S),\vol(\bar{S})\}}.$$
Kannan, Vempala and Veta~\cite{KVV04:clustering} suggest using the conductance as a measure of a set being a general community (in contrast of cyber-communities), since the smaller the conductance it, the more likely that the set is a community with dense intra-connections and sparse inter-connections. Spielman and Teng give the first local clustering algorithm to find subgraphs with small conductance by using the truncated random walk~\cite{ST04:linear,ST08:local}. Anderson, Chung and Lang~\cite{ACL06:localpr}, Anderson and Peres~\cite{AP09:evolvingset}, Kwok and Lau~\cite{KL12:rw} and Oveis Gharan and Trevisan~\cite{OT12:clustering} then give local algorithms for conductance with better approximation ratio or running time. All their local algorithms are based on the Cheeger's inequality that relates the second smallest eigenvalue of $\LL$ to the conductance~\cite{AM85:iso,Alo86:expander,SJ89:counting}, similar to our algorithms which depend on Trevisan's spectral inequality that relates the largest eigenvalue of $\LL$ to the B-ratio.

Some works studied the \textit{small set expander}, that is, to find small set with small conductance. This problem is of interest not only for the reason that it has applications in finding small communities in social networks, but also that it is closely related to the unique games conjecture~\cite{RS10:expansion}. Arora, Barak and Steurer~\cite{ABS10:subexp}, Louis, Raghavendra, Tetali and Vempala~\cite{LRTV12:sparse}, Lee, Oveis Gharan and Trevisan\cite{LOT12:spectral}, Kwok and Lau~\cite{KL12:rw}, Oveis Gharan and Trevisan~\cite{OT12:clustering} and O'Donnell and Witmer~\cite{OW12:sse} have given spectra based approximation algorithms and characterizations of this problem.
The latter three works have recently shown that for any $0<\epsilon<1$, $$\phi(\vol(G)/k^{1-\epsilon})\leq O(\sqrt{\lambda_{k}\log_kn}),$$
where $\phi(k)$ is the \textit{expansion profile} of $G$ and is defined as 
$$\phi(k):=\min_{S:\vol(S)\leq k}\phi(S).$$
Their spectral characterization of the expansion profile as well as the Cheeger's inequality all use the first $k$ smallest eigenvalues of $\LL$, which is comparable to our characterization of the dense bipartite profile by the $k$th largest eigenvalue of $\LL$ as given in inequality~(\ref{eqn:dbspectra}).

Feige, Kortsarz and Peleg~\cite{FKP01:dense}, and Bhaskara et al.~\cite{BCCFV10:density} give non-local algorithms for the densest $k$-subgraphs. Charikar~\cite{Cha00:greedy}, Andersen~\cite{And10:local}, Andersen and Chellapilla~\cite{AC09:dense}, Khuller and Saha\cite{KS09:dense} studied approximation (or local) algorithms for dense subgraphs based on other measures. Arora et al.~\cite{AGSS12:overlap} and
Balcan et al.~\cite{BBBCT13:community} investigate the problem of finding overlapping communities in networks.

\section{Preliminaries}
Let $G=(V,E)$ be an undirected weighted graph and let $n:=|V|$ and $m:=|E|$. Let $d(v)$ denote the weighted degree of vertex $v$. For any vertex subset $S\subseteq V$, let $\bar{S}:=V\backslash S$ denote the complementary of $S$. Let $e(S)$ be the number of edges in $S$ and define the volume of $S$ to be the sum of degree of vertices in $S$, that is $\vol(S):=\sum_{v\in S}d(v)$. Let $\vol(G):=\vol(V)=2m$. For any two subsets $L,R\subseteq V$, let $e(L,R)$ denote the number of edges between $L$ and $R$. For two disjoint subsets $L,R$, that is, $L\cap R=\emptyset$, we will use $U=(L,R)$ to denote subgraph induced on $L$ and $R$, which is also called the pair subgraph. We will also use $U$ to denote $L\cup R$. Given $U=(L,R)$, the \textit{bipartiteness ratio (or B-ratio) of $U$} is defined as
\begin{eqnarray}
\beta(L,R):=\frac{2e(L)+2e(R)+e(U,\bar{U})}{\vol(U)}.\nonumber
\end{eqnarray}

The \textit{B-ratio of a set $S$} is defined to be the minimum value of $\beta(L,R)$ over all the possible partitions $L,R$ of $S$, that is,
$$\beta(S):=\min_{(L,R)\,\textrm{partition of $S$}}\beta(L,R).$$

The B-ratio of the graph $G$ is defined as
$$\beta(G):=\min_{S\subseteq V}\beta(S).$$

We are interested in finding small subgraphs with small B-ratio. In the following, we use lower bold letters to denote vectors. Unless otherwise specified, a vector $\p$ is considered to be a row vector, and $\p^T$ is its transpose. For a vector $\p$ on vertices, let $\supp(\p)$ denote the support of $\p$, that is, the set of vertices on which the $\p$ value is nonzero. Let $\norm{\p}_1$ and $\norm{\p}_2$ denote the $L^1$ and $L^2$ norm of $\p$, respectively. Let $\abs{\p}$ denotes its absolute vector, that is, $\abs{\p}(v):=\abs{\p(v)}$.
For a vector $\p$ and a vertex subset $S$, let $\p(S):=\sum_{v\in S}\p(v)$. For $L,R$, let $\p(L,-R):=\sum_{v\in L}p(v)-\sum_{v\in R}p(v)$. One useful observation is that for any partition $(L,R)$ of $S$, $\p(L,-R)\leq \abs{\p}(S)$. Also note that  there exists a partition $(L_0,R_0)$ of $S$ such that $\p(L_0,-R_0)=\abs{\p}(S)$. Actually, $L_0$ is the set of vertices with positive $\p$ value and $R_0$ is the set of the remaining vertices, that is, $L_0=\{v\in S:\p(v)>0\}$ and $R_0=\{v\in S:\p(v)\leq 0\}$.

For any vertex $v$, let $\chi_v$ denote the indicator vector on $v$.
Now let $A$ denote the adjacency matrix of the graph such that $A_{uv}$ is the weight of edge $u\sim v$. Let $D$ denote the diagonal degree matrix. Define the \textit{random walk matrix $W$}, the \textit{(normalized) Laplacian matrix $\LL$} and the \textit{quasi-Laplacian matrix $M$} of the graph $G$ as
$$W:=D^{-1}A,\LL:=I-D^{-1/2}AD^{-1/2},M:=I-D^{-1}A.$$
It is well known that these three matrices are closely related. In particular, for regular graphs, $\LL$ and $M$ are the same; and if we let $\lambda_0\leq \lambda_1\leq\cdots\leq \lambda_{n-1}$ be the eigenvalues of $\LL$, then $\{1-\lambda_i\}_{0\leq i\leq n-1}$ and $\{\lambda_i\}_{0\leq i\leq n-1}$ are the eigenvalues of $W$ and $M$, respectively. Note that for general graphs, though the eigenvalues of $\LL$ and $M$ are the same, the corresponding eigenvectors might be different. In this paper, we will mainly use the quasi-Laplacian $M$ to give both algorithms and spectral characterization for the small dense bipartite subgraph problem. If we let $\vv_0,\vv_1,\cdots,\vv_{n-1}$ be the corresponding left eigenvectors of $M$, then we have the following spectral inequality given by Trevisan~\cite{Tre09:maxcut} (see also~\cite{Pen12:local} as Trevisan did not explicitly state his result in terms of the matrix $M$).

\begin{lemma}[\cite{Tre09:maxcut}]~\label{lemma:tre}
Let $\beta(G)$, $\lambda_{n-1}$ and $\vv_{n-1}$ defined as above. We have that $$\beta(G)\leq \sqrt{2(2-\lambda_{n-1})}.$$ Furthermore, a pair subgraph $(X,Y)$ with B-ratio at most $\sqrt{2(2-\lambda_{n-1})}$ can be found by a sweep process over $\vv_{n-1}$.
\end{lemma}

The sweep process mentioned above is defined as follows.
\begin{definition}(Sweep process)~\label{def:sweep}
Given a vector $\p$, the \textit{sweep (process)} over $\p$ is defined by performing the following operations:
\begin{enumerate}
\item Order the vertices so that
$$\frac{|\p(v_1)|}{d(v_1)}\geq \frac{|\p(v_2)|}{d(v_2)}\geq\cdots\geq \frac{|\p(v_n)|}{d(v_n)}.$$
\item For each $i\leq n$, let $L_i(\p):=\{v_j:\p(v_j)>0 \, \textrm{and}\, j\leq i\}$, $R_i(\p):=\{v_j:\p(v_j)\leq 0\, \textrm{and}\,j\leq i\}$ and $S_i(\p):=(L_i(\p),R_i(\p))$, which we call the sweep set of the first $i$ vertices. Compute the B-ratio of $S_i(\p)$.
\end{enumerate}
\end{definition}

In Trevisan's inequality, to find the subgraph with small B-ratio, we just need to output the sweep set with the minimum B-ratio over the all the sweep sets. Trevisan also showed the tightness (within constant factors) of this inequality in the sense that there exist graphs such that the two quantities in both hands of the inequality in Lemma~\ref{lemma:tre} are asymptotically the same. The sweep process as well as Trevisan's inequality are the bases of our algorithms for the small dense bipartite-like subgraphs.

We will use the following truncation operator to design local algorithms.
\begin{definition}(Truncation operator)
Given a vector $\p$ and a nonnegative real number $\xi$, we define the $\xi$-truncated vector of $\p$ to be:
\begin{displaymath}
[\p]_\xi(u)=\left\{\begin{array}{ll}
\p(u) & \textrm{if $|\p(u)|\geq \xi d(u)$,} \\
0 & \textrm{otherwise.}
\end{array}\right.
\end{displaymath}
\end{definition}

The following facts are straightforward and will be useful in the remaining proofs.

\begin{fact}~\label{fact:truncation}
For any vector $\p$ and $0\leq\xi\leq 1$,
\begin{enumerate}
\item $\abs{[\p]_\xi}\leq \abs{\p}\leq \abs{[\p]_\xi}+\xi\dd$, where $\dd$ is the degree vector.
\item $\vol(\supp([\p]_\xi))=\sum_{v\in \supp([\p]_\xi)}d(v)\leq \sum_{v\in \supp([\p]_\xi)}|\p(v)|/\xi\leq \norm{\p}_1/\xi$.
\end{enumerate}
\end{fact}

\section{Approximation Algorithm for the Small Dense Bipartite-like Subgraphs}
~\label{sec:profile}
In this section, we first give the description of our approximation algorithm for the small dense bipartite-like subgraph, the main subroutine of which is the sweep process over a set of vectors $\chi_vM^t$. We then introduce a potential function $J(\p,x)$ and give both upper bound and lower bound of the potential function $J(\chi_vM^t,x)$ under certain conditions, using which we are able to show the correctness of our algorithm and thus prove Theorem~\ref{thm:smallBR}.

Now we describe our algorithm \verb|SwpDB| (short for ``sweep for dense bipartite'') for finding the small dense bipartite-like subgraphs.
\begin{center}
\begin{tabular}{|p{\textwidth}|}
\hline
\verb|SwpDB|$(G,k,\theta,\epsilon)$\\
\hline

Input: A graph $G$, a target volume $k>4$, a target B-ratio $\theta$, an error parameter $0<\epsilon<1/2$.

Output: A subgraph $(X,Y)$.
\begin{enumerate}
\item Let $T=\frac{\epsilon\ln 2k}{2\theta}$. Let $K=2k^{1+\epsilon}$.
\item Sweep over all vectors $\chi_vM^t$, for each vertex $v\in V$ and $t\leq T$, to obtain a family $\mathcal{F}$ of sweep sets with volume at most $K$.
\item Output the subgraph $(X,Y)$ with the smallest B-ratio ratio among all sets in $\mathcal{F}$.
\end{enumerate}\\
\hline
\end{tabular}
\end{center}

\subsection{A Potential Function}
We define a potential function $J:[0,2m]\rightarrow \mathbb{R}^+$:
$$J({\bf p},x):=\max_{\substack{\w\in [0,1]^n\\ \sum_{v\in V}\w(v)d(v)=x}}\sum_{v\in V}|{\bf p}(v)|\w(v).$$
Note that our potential function is similar to a potential function for bounding the convergence of $\p(\frac{I+W}{2})^t$ in terms of the conductance given by Lov{\'a}sz and Simonovits~\cite{LS90:mixing,LS93:random}. Here we will use $J({\bf p},x)$ to bound the convergence of $\q M^t$ in terms of the B-ratio of the sweep sets.

There are two useful ways to see this potential function.
\begin{enumerate}
\item We view each edge $u\sim v\in E$ as two directed edges $u\rightarrow v$ and $v\rightarrow u$. For each directed edge $e=u\rightarrow v$, let $\p(e)=\frac{\p(u)}{d(u)}$. Order the edges so that
    $$|\p(e_1)|\geq |\p(e_2)|\geq\cdots \geq |\p(e_{2m})|.$$

Now we can see that for an integer $x$, $J(\p,x)=\sum_{j=1}^x|\p(e_j)|$. For other fractional $x=\lfloor x\rfloor+r$, $J(\p,x) =(1-r)J(\p,\lfloor x\rfloor)+rJ(\p,\lceil x\rceil)$.

Also it is easy to see that for any directed edge set $F$, $|\p|(F):=\sum_{e\in F}|\p|(e)\leq J(\p,|F|)$, since the former is a sum of $|\p|$ values of one specific set of edges with $|F|$ edges and the latter is the maximum over all such possible edge sets.
\item Another way to view the potential function is to use the sweep process over $\p$ as in Definition~\ref{def:sweep}. By the definitions of the potential function and the sweep process, we have the following observations.
\begin{enumerate}
\item For $x=\vol(S_i(\p))$, $J(\p,x)=\sum_{j=1}^i|\p(v_j)|=|\p|(S_i(\p))
=\p(L_i(\p),-R_i(\p))$. And $J(\p,x)$ is linear in other values of $x$.
\item For any set $S$, $\abs{\p}(S)\leq J(\p,\vol(S))$, since the former is the sum of $|\p(v)|/d(v)$ values of vertices in $S$ and the latter is the maximum sum over all sets with $|S|$ vertices;
\end{enumerate}
\end{enumerate}

From both views, we can easily see that the potential function is a non-decreasing and concave function of $x$.

\subsection{An Upper Bound for the Potential Function}
Now we upper bound $J(\p M,x)$ in terms of $J(\p,x')$ and the B-ratio of the sweep set of $\p M$.
\begin{lemma}[Convergence Lemma]~\label{lemma:converge}
For an arbitrary vector $\p$ on vertices, if $\beta(L_i(\p),R_i(\p))\geq \Theta$, then for $x=\vol(S_i(\p))$, we have
$$J(\p M,x)\leq J(\p, x+\Theta x)+J(\p, x-\Theta x).$$
\end{lemma}

We remark that the proof heavily depends on the definition and combinatorial property of the B-ratio of a set. Though the form is similar to the corresponding characterization of conductance given by Lov{\'a}sz and Simonovits, the two proofs are very different.
\begin{proof}[Proof of Lemma~\ref{lemma:converge}]
We show that for any $U=(L,R)$, we have that
\begin{eqnarray}
\p M(L,-R)\leq J(\p, \vol(U)(1+\beta(L,R)))+J(\p, \vol(U)(1-\beta(L,R))). \label{eqn:potential}
\end{eqnarray}
Then the lemma follows by letting $U=S_i(\p)=(L_i(\p),R_i(\p))$ and that
\begin{eqnarray*}
J(\p M,x)&=&\p M(L_i(\p),-R_i(\p))\\
&\leq& J(\p,x(1+\beta(L_i(\p),R_i(\p))))+J(\p,x(1-\beta(L_i(\p),R_i(\p))))\\
&\leq&J(\p,x(1+\Theta))+J(\p,x(1-\Theta)),
\end{eqnarray*}
where the last inequality follows from the concavity of $J(\p,x)$.

Now we show inequality~(\ref{eqn:potential}). Let $L_1\rightarrow L_2$ denote the set of direct edges from $L_1$ to $L_2$ for two arbitrary vertex sets $L_1$ and $L_2$. We have that
\begin{eqnarray*}
\p M(L,-R)&=&\p(I-D^{-1}A)(L,-R)\\
&=&\p(L)-\p(R)-\p D^{-1}A(L)+\p D^{-1}A(R)\\
&=&\sum_{v\in L}\sum_{v\rightarrow u}\frac{\p(v)}{d(v)}-\sum_{v\in R}\sum_{v\rightarrow u}\frac{\p(v)}{d(v)}-\sum_{v\in L}\sum_{u\rightarrow v}\frac{\p(u)}{d(u)}+\sum_{v\in R}\sum_{u\rightarrow v}\frac{\p(u)}{d(u)}\\
&=&\sum_{e\in L\rightarrow \bar{L}}\p(e)-\sum_{e\in R\rightarrow \bar{R}}\p(e)-
\sum_{e\in R\rightarrow L}\p(e)
-\sum_{e\in \bar{U}\rightarrow L}\p(e)\\
&&+
\sum_{e\in L\rightarrow R}\p(e)+
\sum_{e\in \bar{U}\rightarrow R}\p(e)\\
&\leq& \sum_{e\in (L\rightarrow \bar{L})\cup (R\rightarrow \bar{R})\cup (\bar{U}\rightarrow U)}|\p(e)| +
\sum_{e\in (L\rightarrow R) \cup (R\rightarrow L)}|\p(e)|\\
&\leq& J(\p,2e(L,R)+2e(U,\bar{U}))
+J(\p,2e(L,R))\\
&\leq & J(\p, \vol(U)+2e(L)+2e(R)+e(U,\bar{U}))\\
&&+J(\p, \vol(U)-2e(L)-2e(R)-e(U,\bar{U})),
\end{eqnarray*}
where the second to last inequality follows from the fact that $|(L\rightarrow \bar{L})\cup (R\rightarrow \bar{R})\cup (\bar{U}\rightarrow U)|=2e(L,R)+2e(U,\bar{U})$, that $|(L\rightarrow R) \cup (R\rightarrow L)|=2e(L,R)$ and that $|\p|(F)\leq J(\p,|F|)$ for an arbitrary (directed) edge set $F$; and the last inequality follows from that $J(\p,x)$ is non-decreasing.
\end{proof}

Now we can use the convergence lemma to upper bound $J(\chi_vM^t,x)$.

\begin{lemma}~\label{lemma:upper}
For any vertex $v\in V$, let $\q_t=\chi_vM^t$, if for all $t\leq T$ and all sweep sets $S_i(\q_t)=(L_i(\q_t),R_i(\q_t))$ of volume at most $K$ have B-ratio at least $\Theta$, that is, $\beta(L_i(\q_t),R_i(\q_t))\geq \Theta$, then for any $t\leq T$, we have $J(q_t,x)\leq \frac{2^tx}{K}+\sqrt{\frac{x}{d(v)}}\Big(2-\frac{\Theta^2}{4}\Big)^t.$
\end{lemma}

\begin{proof}
The proof is by induction and is similar to the Lemma 4.2 in~\cite{OT12:clustering}.

If $t=0$, then the LHS is $x/d(v)$ for $x\leq d(v)$ and is $1$ for $x>d(v)$, and the RHS is at least $\sqrt{x/d(v)}$ for any $x\in [0,2m]$. Thus, the lemma holds in this case.

Assume the lemma holds for $t-1$. Since $J(q_t,x)$ is piecewise linear in $x$, and the RHS is concave, we only need to show the lemma holds for $x=\vol(S_i(q_t))$ for any $i\leq n$.
\begin{itemize}
\item For $x> K$, the RHS is at least $2^t$. On the other hand, for any vector $\p$, we have
    \begin{eqnarray*}
    J(\p M,2m)=
    \norm{\p M}_1=\sum_{u}\abs{\sum_{v}\p(v)M_{vu}}&\leq& \sum_{v}|\p(v)|\sum_{u}\abs{M_{vu}}\\
    &\leq& 2\norm{\p}_1=2J(\p,2m),
    \end{eqnarray*}
    Therefore,
    $$J(\q_t,x)\leq J(\q_t,2m)\leq 2J(\q_{t-1},2m)\leq\cdots\leq 2^tJ(\q_0,2m)=2^t.$$
    So the lemma holds for $x$ in this case.
\item For $x\leq K$, recall that $x=\vol(S_i(q_t))$, by Lemma~\ref{lemma:converge} and the induction hypothesis, we have
    \begin{eqnarray*}
    J(\q_t,x)&\leq& J(\q_{t-1},x+x\Theta)+J(\q_{t-1},x-x\Theta)\\
    &\leq& \frac{2*2^{t-1}x}{K}+\sqrt{\frac{x}{d(v)}}
    \Big(2-\frac{\Theta^2}{4}\Big)^{t-1}(\sqrt{1+\Theta}+\sqrt{1-\Theta})\\
    &\leq& \frac{2^tx}{K}+\sqrt{\frac{x}{d(v)}}\Big(2-\frac{\Theta^2}{4}\Big)^t,
    \end{eqnarray*}
    where the last inequality follows from that $$\sqrt{1+\Theta}+\sqrt{1-\Theta}\leq 2-\frac{\Theta^2}{4}.$$
\end{itemize}
This completes the proof.
\end{proof}

\subsection{A Lower Bound for the Potential Function}
We show that if the graph contains a pair subgraph with small B-ratio, then we can have a good lower bound on $J(\chi_vM^t)$ for some vertex $v$. The following lemma is similar to the upper bounds on the escaping probability of random walks given by Oveis Gharan and Trevisan~\cite{OT12:clustering}. Our proof uses a new spectral analysis by using the \textit{$H$-norm} of a vector and may be modified to give a different proof for the corresponding results in~\cite{OT12:clustering}.

\begin{lemma}~\label{lemma:lower}
If $U=(L,R)$ has B-ratio $\beta(L,R)\leq \theta$, then for any integer $t>0$,
\begin{enumerate}
\item there exists a vertex $v\in U$ such that  $\abs{\q_t}(U)\geq (2-2\theta)^t$, where $\q_t=\chi_vM^t$;
\item there exists a subset $U^t\subseteq U$ with $\vol(U^t)\geq \vol(U)/2$ satisfying that for any $v\in U^t$ and $\q_t=\chi_vM^t$,
$J(\q_t,\vol(U))\geq |\q_t|(U)\geq \frac{1}{400}(2-6\theta)^t$,
where we have assumed that $\theta<1/3$.
\end{enumerate}
\end{lemma}

\begin{proof} To give the proof, we will use the following notation. For a set $U=(L,R)$, define $\rho_{U}$ and $\psi_{U}$ as
\begin{equation*}
\begin{aligned}[c]
\rho_{U}(v)=\left\{ \begin{array}{ll}
d(v)/\vol(U) & \textrm{if $v\in L$,}\\
-d(v)/\vol(U) & \textrm{if $v\in R$,}\\
0 & \textrm{otherwise.}
\end{array}
\right.
\end{aligned}
\qquad
\begin{aligned}[c]
\psi_{U}(v)=\left\{ \begin{array}{ll}
\sqrt{d(v)/\vol(U)} & \textrm{if $v\in L$,}\\
-\sqrt{d(v)/\vol(U)} & \textrm{if $v\in R$,}\\
0 & \textrm{otherwise.}
\end{array}
\right.
\end{aligned}
\end{equation*}

\begin{enumerate}
\item For the first part, we will show that
\begin{eqnarray}
\rho_UM^t(L,-R)\geq (2-2\theta)^t.~\label{eqn:lower}
\end{eqnarray}
If the above inequality holds, then by the fact that $$\rho_UM^t(L,-R)=\sum_{v\in U}\frac{d(v)}{\vol(U)}\sgn(v,L)\chi_vM^t(L,-R),$$
where $\sgn(v,L)$ equals $1$ if $v\in L$ and $-1$ if $v\in R$, we know there exists a vertex $v\in U$ satisfying $\sgn(v,L)\chi_vM^t(L,-R)\geq (2-2\theta)^t$. Then the lemma follows from the fact that $|\p|(U)\geq \max\{\p(L,-R),\p(R,-L)\}$ for any $\p$.

To show inequality~(\ref{eqn:lower}), we note that for any $t\geq 0$,
$$\rho_UM^t(L,-R)=\rho_UD^{-1/2}\LL^t D^{1/2}(L,-R)=\psi_U\LL^t\psi_U^T.$$
On the other hand,
\begin{eqnarray*}
\psi_U(2I-\LL)\psi_U^T&=&\psi_UD^{-1/2}(D+A)D^{-1/2}\psi_U^T\\
&=&\sum_{u\sim v}(\psi_U(u)/\sqrt{d(u)}+\psi_U(v)/\sqrt{d(v)})^2\\
&=&\frac{4e(L)+4e(R)+e(U,\bar{U})}{\vol(U)}
\leq 2\theta,
\end{eqnarray*}
which implies that
\begin{eqnarray}
\psi_U\LL\psi_U^T\geq 2-2\theta. \label{eqn:base}
\end{eqnarray}

Now recall that $0=\lambda_0\leq\lambda_1\leq\cdots\leq \lambda_{n-1}\leq 2$ are the eigenvalues of the Laplacian $\LL$. Let $\vv'_0,\vv'_1,\cdots,\vv'_{n-1}$ be the corresponding orthonormal eigenvectors of $\LL$. If we write $\psi_U=\sum_{i}\alpha_i\vv'_i$, then by inequality~(\ref{eqn:base}), we have
$\sum_{i}\lambda_i\alpha_i^2\geq 2-2\theta$. Therefore,
$$\psi_U\LL^t\psi_U^T= \sum_{i}\lambda_{i}^t\alpha_i^2\geq (\sum_{i}\lambda_{i}\alpha_i^2)^t\geq (2-2\theta)^t,$$
where the second inequality follows from the fact that $\sum_i\alpha_i^2=\norm{\psi_U}_2^2=1$ and the Chebyshev's sum inequality.

\item For the second part, we show that for any set $Z=(L_Z,R_Z)$ such that $L_Z\subseteq L$, $R_Z\subseteq R$ and $\vol(Z)\geq\frac{\vol(U)}{2}$,
\begin{eqnarray}
\rho_ZM^t(L_Z,-R_Z)\geq \frac{1}{400}(2-6\theta)^t,~\label{eqn:half}
\end{eqnarray}
from which we know there exists at least one vertex $v$ in $Z$ such that $$|\chi_vM^t|(U)\geq |\chi_vM^t|(Z)\geq \sgn(v,L_Z)\chi_vM^t(L_Z,-R_Z)\geq \frac{1}{400}(2-6\theta)^t.$$
Then by the choice of $Z$, we know that the set $U^t:=\{v: |\chi_vM^t|(U)\geq \frac{1}{400}(2-6\theta)^t\}$ has volume at least $\vol(U)/2$ and the lemma's statement holds.

On the other hand, we have that $\rho_ZM^t(L_Z,-R_Z)=\psi_ZM^t\psi_Z^T$ for the same reason as in the first part of the proof, so we only need to show that
$$\psi_ZM^t\psi_Z^T\geq \frac{1}{400}(2-6\theta)^t.$$

Let $H=\{i:\lambda_i\geq 2-6\theta\}$. For an vector $\p$, define its $H$-norm as $\Vert \p \Vert_H:=\sqrt{\sum_{i\in H}\langle \p,\vv'_i \rangle ^2}$. It is straightforward to show that $\norm{\cdot}_H$ is a seminorm. Recall that $\psi_U=\sum_{i}\alpha_i\vv'_i$ and $\sum_i\lambda_i\alpha_i^2\geq 2-2\theta$. By the definition of $H$-norm and that $\norm{\psi_U}_2^2=1$, we have
$$\sum_i\lambda_i\alpha_i^2\leq 2\sum_{i\in H}\alpha_i^2+(2-6\theta)\sum_{i\notin H}\alpha_i^2=2\Vert\psi_U\Vert_H^2+(2-6\theta)(1-\Vert\psi_U\Vert_H^2),$$
which gives that
$$\Vert\psi_U\Vert_H^2\geq 2/3.$$

Now we write $\psi_Z=\sum_{i}\beta_i\vv'_i$. It is easy to show that
\begin{eqnarray*}
\norm{\psi_U-\psi_Z}_2^2&=&\sum_{v\in Z}\Big(\sqrt{\frac{d(v)}{\vol(Z)}}-\sqrt{\frac{d(v)}{\vol(U)}}\Big)^2+\sum_{v\in U\backslash Z}\frac{d(v)}{\vol(U)}\\
&=&\sum_{v\in Z}d(v)\Big(\frac{1}{\vol(Z)}-\frac{2}{\sqrt{\vol(Z)\vol(U)}}+\frac{1}{\vol(U)}\Big)+\frac{\vol(U\backslash Z)}{\vol(U)}\\
&=&2-2\sqrt{\frac{\vol(Z)}{\vol(U)}}\\
&\leq& 2-\sqrt{2},
\end{eqnarray*}
where the last inequality follows from our assumption that $\vol(Z)\geq \vol(U)/2$.

Hence, $$\norm{\psi_U-\psi_Z}_H\leq\norm{\psi_U-\psi_Z}_2\leq \sqrt{2-\sqrt{2}}.$$

Then by the triangle inequality, we have
$$\norm{\psi_Z}_H\geq \norm{\psi_U}_H-\norm{\psi_U-\psi_Z}_H\geq \sqrt{\frac{2}{3}}-\sqrt{2-\sqrt{2}}> \frac{1}{20}.$$

Finally, we have
$$\psi_Z\LL^t\psi_Z=\sum_{i}\lambda_i^t\beta_i^2\geq(2-6\theta)^t\norm{\psi_Z}_H^2> \frac{1}{400}(2-6\theta)^2.$$
\end{enumerate}
\end{proof}

Now we are ready to prove Theorem~\ref{thm:smallBR}.
\begin{proof}[Proof of Theorem~\ref{thm:smallBR}]
Clearly the algorithm \verb|SwpDB| runs in polynomial time. Now we show the correctness of the algorithm. Let $\Theta=4\sqrt{\theta/\epsilon}$. Assume on the contrary that the algorithm \verb|SwpDB|$(G,k,\theta,\epsilon)$ does not find a desired subgraph, and thus for any $v\in V$, and $t\leq T=\frac{\epsilon\ln 2k}{2\theta}$, the sweep sets $S_i(\chi_vM^t)$ of volume at most $K=2k^{1+\epsilon}$ have B-ratio at least $4\sqrt{\theta/\epsilon}$.
Then by Lemma~\ref{lemma:upper}, for any $v\in V$,
\begin{eqnarray*}
J(\chi_vM^T,k)\leq 2^T\frac{k}{2k^{1+\epsilon}}+\sqrt{k}\Big(2-\frac{\Theta^2}{4}\Big)^T&\leq& 2^T \Big(\frac{1}{2k^{\epsilon}}+\sqrt{k}\Big(1-\frac{2\theta}{\epsilon}\Big)^{\frac{\epsilon\ln 2k}{2\theta}}\Big)\\
&\leq& 2^T\Big(\frac{1}{2k^{\epsilon}}+\frac{1}{2k^{1/2}}\Big)\\
&\leq& 2^Tk^{-\epsilon},
\end{eqnarray*}
where in the last inequality we used the assumption that $0<\epsilon<1/2$.

On the other hand, since $U=(L,R)$ is subgraph such that $\beta(L,R)\leq \theta$ and $\vol(U)\leq k$, then by Lemma~\ref{lemma:lower}, we know that there exists a vertex $u\in U$ such that,
\begin{eqnarray*}J(\chi_uM^T,k)\geq (2-2\theta)^T\geq 2^T(1-\theta)^{\frac{\epsilon\ln 2k}{2\theta}}\geq 2^T e^{-\frac{\epsilon \ln 2k}{2(1-\theta)}}&=&2^T(2k)^{-\frac{\epsilon}{2(1-\theta)}}\\
&>&2^Tk^{-\epsilon},
\end{eqnarray*}
where in the last inequality we used the assumption that $\theta<1/4$ and that $k>4$. Hence we have derived a contradiction, which completes the proof.
\end{proof}
By using a simple trace bound, we can obtain the following corollary that gives a spectral characterization of the small dense bipartite-like subgraphs and thus establish inequality~(\ref{eqn:dbspectra}).

\begin{corollary}~\label{thm:spectra}
If $\lambda_{n-k}\geq 2-2\eta$, then there is a polynomial time algorithm such that for any $0<\epsilon<1$, it finds a subset $(X,Y)$ of volume at most $O(\vol(G)/k^{1-\epsilon})$ and B-ratio $O(\sqrt{16(\eta/\epsilon)\log_kn})$.
\end{corollary}

\begin{proof}
Given $k,\eta,\epsilon$, we set $T=\frac{\epsilon\ln k}{2\eta}$, $K=\frac{\vol(G)}{0.5k^{1-\epsilon}}$, and run the step 2 and 3 of the algorithm \verb|SwpDB| to find a subgraph, which clearly runs in polynomial time. Assume that during this process, all the sweep sets $S_i(\chi_vM^t)$ of volume at most $K$ have B-ratio $\Theta=\sqrt{16(\eta/\epsilon)\log_kn}$, for any $v\in V$ and $t\leq T$. Then, by Lemma~\ref{lemma:upper}, we have that for any $v\in V$,
$$\chi_vM^T\chi_v^T\leq J(\chi_vM_T,d(v))\leq 2^T\frac{d(v)}{K}+\Big(2-\frac{\Theta^2}{4}\Big)^T.$$
Therefore,
\begin{eqnarray*}
\sum_{v\in V}\chi_vM^T\chi_v^T&\leq & 2^T\Big(\frac{\vol(G)}{K}+n\Big(1-\frac{\Theta^2}{8}\Big)^T\Big)\\
&=& 2^T\Big(0.5k^{1-\epsilon}+n\Big(1-\frac{2\eta\log_kn}{\epsilon}\Big)^{\frac{\epsilon\ln k}{2\eta}}\Big)\\
&\leq &2^T(0.5k^{1-\epsilon}+1)\\
&<& 2^Tk^{1-\epsilon}.
\end{eqnarray*}
On the other hand, by the trace formula,
$$\sum_{v\in V}\chi_vM^T\chi_v^T=\Tr(M^T)=\sum_{i=1}^{n}\lambda_i^t\geq k(2-2\eta)^T=2^Tk(1-\eta)^{\frac{\epsilon\ln k}{2\eta}}\geq 2^Tk^{1-\epsilon},$$
which is a contradiction.
\end{proof}

\section{A Local Algorithm for Dense Bipartite-like Subgraphs}~\label{sec:local}
We will use the truncated operation to give our local algorithm \verb|LocDB| (short for ``local algorithm for dense bipartite subgraph'').
\begin{center}
\begin{tabular}{|p{\textwidth}|}
\hline
\verb|LocDB|$(G,v,k,\theta,\epsilon)$\\
\hline
Input: A graph $G$, a vertex $v$, a target volume $k>2560000$, a target B-ratio $\theta<1/3$ and an error parameter $0<\epsilon<1/2$.

Output: A subgraph $(X,Y)$.
\begin{enumerate}
\item Let $T=\frac{\epsilon\ln 1600k}{6\theta}$. Let $\xi_0 =\frac{k^{-1-\epsilon}}{800T}$, $\xi_t=\xi_0 2^t$. Let $\qq_0:=\chi_v$, $\rr_0:=[\qq_0]_{\xi_0}$. Let $\mathcal{F}=\emptyset$.
\item For each time $1\leq t\leq T$:
\begin{enumerate}
\item Compute $\qq_{t}:=\rr_{t-1}M$, $\rr_{t}:=[\qq_{t}]_{\xi_t}$;
\item Sweep over the support of $\qq_t$ and add to $\mathcal{F}$ all the sweep sets.
\end{enumerate}
\item Output the subgraph $(X,Y)$ with the smallest B-ratio ratio among all sets in $\mathcal{F}$.
\end{enumerate}\\
\hline
\end{tabular}
\end{center}

Note that in the algorithm we just sweep the \textit{support} of a given vector, which is important for the computation to be local.

Inspired by the proof of the correctness of \texttt{SwpDB}, we will use the upper bound and lower bound of the potential function $J(\qq_t,x)$ to show the correctness of the local algorithm. Such bounds can be obtained by combining the following properties of the truncation operations in the algorithm.

\begin{proposition}~\label{prop:trunc}
For any vertex $v$, if $\q_t=\chi_vM^t$ and $\qq_t,\rr_t$ are as defined in the algorithm \verb|LocDB|, then for any $t\geq 0$,
\begin{enumerate}
\item $\norm{\qq_t}_1\leq 2^t$;
\item $\abs{\rr_t-\q_t}\leq \xi_0 t2^t\dd$, where $\dd$ is the degree vector.
\end{enumerate}
\end{proposition}
\begin{proof}
We prove both the inequalities by induction.
\begin{enumerate}
\item If $t=0$, the inequality trivially holds since $\qq_0=\chi_v$. Now assume that the inequality holds for $t-1$. Then
    $$\norm{\qq_t}_1=\norm{r_{t-1}M}_1=\norm{[\qq_{t-1}]_{\xi_{t-1}}M}_1\leq \norm{[\qq_{t-1}]_{\xi_{t-1}}}_1*2\leq 2\norm{\qq_{t-1}}_1\leq 2^t,$$
    where the third inequality follows by the fact that $\norm{\p M}_1\leq 2\norm{\p}_1$ for all $\p$; the fourth inequality follows by the definition of truncation; and the last inequality follows by the induction.
\item If $t=0$, the inequality holds since $\q_0=\rr_0=[\q_0]_{\xi_0}=\chi_v$. If $t=1$, then $\rr_1=[\qq_1]_{\xi_1}=[\rr_0M]_{\xi_1}=[\q_0M]_{\xi_1}=[\q_1]_{\xi_1}$, and thus $\abs{\rr_1-\q_1}\leq \xi_1\dd=2\xi_0\dd$ by the Fact~\ref{fact:truncation}. Now assume that the inequality holds for $t-1$, that is,
    $\abs{\rr_{t-1}-\q_{t-1}}\leq \xi_0 (t-1)2^{t-1}\dd$, which is equivalent to $\abs{(\rr_{t-1}-\q_{t-1})D^{-1}}\leq \xi_0 (t-1)2^{t-1}\1$, where $\1$ is the all $1$ vector. On the other hand,
\begin{eqnarray*}
\abs{\rr_t-\q_t}=\abs{[\rr_{t-1}M]_{\xi_t}-\q_t}&\leq& \abs{\rr_{t-1}M-\q_t}+\xi_t\dd\\
&=&\abs{(\rr_{t-1}-\q_{t-1})D^{-1}(D-A)}+\xi_t\dd\\
&\leq&2*\xi_0 (t-1)2^{t-1}\dd+\xi_0 2^t\dd\\
&=&\xi_0 t2^t\dd,
\end{eqnarray*}
where the second to last inequality follows from the induction hypothesis and the fact that for any vector $\p$, if $\abs{\p}\leq c\1$ for some constant $c$, then for any vertex $v$, $$\abs{\p(D-A)(v)}=\abs{\sum_{u}\p(u)(D_{vu}-A_{vu})}\leq \sum_{u}\abs{\p(u)}(D_{vu}+A_{vu})\leq 2cd(v).$$
\end{enumerate}
\end{proof}

Note that the second part of Proposition~\ref{prop:trunc} directly implies a lower bound on $J(\qq_t,x)$. More specifically, we have the following corollary.

\begin{corollary}~\label{lemma:locallower}
For any set $U$, $\abs{\qq_t}(U)\geq \abs{\rr_t}(U)\geq \abs{\q_t}(U)-\xi_0 t2^t\vol(U)$.
\end{corollary}

We can also give an upper bound on $J(\qq_t,x)$.
\begin{lemma}~\label{lemma:localupper}
For any vertex $v$, $T>0$, $\Theta<1$, if for any $t\leq T$, the sweep sets $S_i(\qq_t)$ of volume at most $K$ have B-ratio at least $\Theta$, then for any $0\leq t\leq T$ and $0\leq x\leq 2m$,
$$J(\qq_t,x)\leq \frac{2^tx}{K}+\sqrt{\frac{x}{d(v)}}\Big(2-\frac{\Theta^2}{4}\Big)^t.$$
\end{lemma}

\begin{proof}
We prove the lemma by combining the following observations and the proof of Lemma~\ref{lemma:upper}.

First we note that for any $t\leq T$ and $x\leq 2m$, $J(\rr_t,x)\leq J(\qq_t,x)$. This follows by the definition of the potential function. More specifically, let $\w\in [0,1]^n$ be a vector that achieves $J(\rr_t,x)$, that is, $\sum_{u}\w(u)d(u)=x$ and $J(\rr_t,x)=\sum_{v}|\rr_t|(v)\w(v)$. Then $J(\rr_t,x)\leq \sum_{v}|\qq_t|(v)\w(v)\leq J(\qq_t,x)$ since for any $v$, $|\rr_t|(v)\leq |\qq_t|(v)$. Furthermore, by the relation between $\qq_t$ and $\rr_{t-1}M$, we can always guarantee that $S_i(\qq_t)=S_i(\rr_{t-1}M)$ for every $i\leq n$.

Then by the conditions given in the lemma and the convergence Lemma~\ref{lemma:converge}, for $x=\vol(S_i(\qq_t))$, we have
\begin{eqnarray}
J(\qq_t,x)=J(\rr_{t-1}M,x)&\leq& J(\rr_{t-1}, x+\Theta x)+J(\rr_{t-1}, x-\Theta x) \nonumber \\
&\leq&J(\qq_{t-1}, x+\Theta x)+J(\qq_{t-1}, x-\Theta x).
\end{eqnarray}

Finally, we can use the same induction as in the proof of Lemma~\ref{lemma:upper} to show that the lemma's statement holds.
\end{proof}
Now by using Corollary~\ref{lemma:locallower} and Lemma~\ref{lemma:localupper}, we can show the correctness of the algorithm \texttt{LocDB} and thus prove Theorem~\ref{thm:local}.

\begin{proof}[Proof of Theorem~\ref{thm:local}]
We first show the correctness of \texttt{LocDB} and then bound its running time.
\begin{itemize}
\item (Correctness.)
As stated in the algorithm, we choose $T=\frac{\epsilon\ln 1600k}{6\theta}$. Let $U_\theta=U^T\subseteq U$ be the subset as described in Lemma~\ref{lemma:lower}, which has volume at least $\vol(U)/2$. Now let $v\in U_\theta$ and assume that in the algorithm \verb|LocDB|$(G,v,k,\theta,\epsilon)$, for any $t\leq T$, all the sweep sets $S_i(\qq_t)$ of volume at most $1600k^{1+\epsilon}$ have B-ratio at least $\Theta=\sqrt{48\theta/\epsilon}$, then by Lemma~\ref{lemma:localupper}, we have
    \begin{eqnarray*}
    J(\qq_t,\vol(S))\leq J(\qq_t,k) &\leq& 2^t\Big(\frac{k}{1600k^{1+\epsilon}}+\sqrt{k}\Big(1-\frac{\Theta^2}{8}\Big)^T\Big)\\
    &\leq&2^t\Big(\frac{1}{1600k^{\epsilon}}+\sqrt{k}\Big(1-\frac{6\theta}{\epsilon}\Big)^{\frac{\epsilon\ln 1600k}{6\theta}}\Big)\\
    &\leq& 2^T\Big(\frac{1}{1600k^{\epsilon}}+\frac{1}{1600k^{1/2}}\Big)\\
    &<& 2^T\frac{k^{-\epsilon}}{800},
    \end{eqnarray*}
    where the last inequality follows from the fact that $0<\epsilon<1/2$.

    On the other hand, by~Lemma~\ref{lemma:lower} and~Corollary~\ref{lemma:locallower} and that $\xi_0 T=\frac{k^{-1-\epsilon}}{800}$, we have
    \begin{eqnarray*}
    \abs{\qq_T}(U)\geq\abs{\chi_vM^T}|U|-\xi_0 T2^T\vol(U)&\geq& 2^T\Big(\frac{1}{400}(1-3\theta)^T-\xi_0 Tk\Big)\\
    &\geq & 2^T\Big(\frac{1}{400}(1-3\theta)^{\frac{\epsilon\ln (1600k)}{2(1-3\theta)}}-\frac{k^{-\epsilon}}{800}\Big)\\
    &\geq&2^T\Big(\frac{1}{400}(1600k)^{-2\epsilon/3}-\frac{k^{-\epsilon}}{800}\Big)\\
    &> &2^T\Big(\frac{1}{400}k^{-\epsilon}-\frac{k^{-\epsilon}}{800}\Big)\\
    &=&\frac{2^Tk^{-\epsilon}}{800},
    \end{eqnarray*}
where the last to third inequality follows from the assumption that $\theta<1/12$ and the last to second inequality follows from the assumption that $k>2560000$. (Also note that we can choose other bounds of $\theta,k$ so long as these two inequalities are satisfied. For example, $\theta<0.03,k>11000$.) Hence, we have derived a contradiction. Therefore, there exists at least one sweep set of volume at most $O(k^{1+\epsilon})$ and B-ratio at most $O(\sqrt{\theta/\epsilon})$.

\item (Running time.)
We first bound the time required in each iteration. For any $t\leq T$, instead of perform the dense vector multiplication to compute $\qq_t$, we keep record of the support of $\rr_t$, which has volume at most
$\norm{\qq_t}_1/\xi_t\leq 2^t/(\xi_02^t)=\xi_0^{-1}.$
    By definition, both the volume of the support and the computational time of $\qq_{t+1}$ are proportional to
    $\vol(\supp(\rr_{t}))$, which is at most $\xi_0^{-1}$ by the property of truncation operation.

    During the sweep process, we only need to sweep the vertices in $\supp(\rr_t)$. Sorting these vertices requires time
    $$O(|\supp(\rr_t)|\ln |\supp(\rr_t)|)\leq O(\vol(\supp(\rr_t))\ln \vol(\supp(\rr_t))).$$
    Computing the B-ratio of the sweep sets requires time $O(\vol(\supp(\rr_{t})))$. Therefore, in a single iteration, the computation takes time
    \begin{eqnarray*}
    O(\vol(\supp(\rr_t))+\vol(\supp(\rr_t))\ln \vol(\supp(\rr_t)))
    &=&O(\xi_0^{-1}+\xi_0^{-1}\ln\xi_0^{-1})\\
    &=&O(\xi_0^{-1}\ln\xi_0^{-1}).
    \end{eqnarray*}

    Since the algorithm takes $T$ iterations, the total running time is thus bounded by $$O(T\xi_0^{-1}\ln\xi_0^{-1})=O(\epsilon^2k^{1+\epsilon}\ln^3k/\theta^{2}).$$
\end{itemize}
\end{proof}

{\large \textbf{Acknowledgements.}}

The research is partially supported by NSFC distinguished young investigator award number
60325206, and its matching fund from the Hundred-Talent Program of the Chinese Academy of Sciences. Both
authors are partially supported by the Grand Project ``Network Algorithms and Digital Information'' of the
Institute of software, Chinese Academy of Sciences. The second author acknowledges the support of ERC grant No.
307696 and NSFC 61003030.

\bibliographystyle{splncs03}

\end{document}